\documentclass[a4paper,fleqn]{article}
\usepackage{fullpage}
\usepackage[hidelinks]{hyperref}
\newcommand*{\email}[1]{\href{mailto:#1}{\nolinkurl{#1}} }

\title{Approximating the Minimal Lookahead\\ Needed to Win Infinite Games}
\author{
\textbf{Martin Zimmermann}\\ Aalborg University, Aalborg, Denmark\thanks{This work was prepared while the author was affiliated with the University of Liverpool.}\\ \email{mzi@cs.aau.dk}
}

\date{\vspace{-4ex}}

\usepackage{amsmath}
\usepackage{amssymb}
\usepackage{amsthm}
\usepackage[utf8]{inputenc}

\newtheorem{proposition}[]{Proposition}
\newtheorem{remark}[]{Remark}
\newtheorem{theorem}[]{Theorem}
\newtheorem{lemma}[]{Lemma}

\usepackage{algorithm}
\usepackage[noend]{algorithmic}

\newcommand{\myquot}[1]{``#1''}

\newcommand{\nats}{\mathbb{N}}

\renewcommand{\epsilon}{\varepsilon}

\newcommand{\set}[1]{\{#1\}}
\newcommand{\pow}[1]{2^{#1}}

\newcommand{\aut}{\mathfrak{A}}

\newcommand{\C}{\mathcal{T}}
\renewcommand{\P}{\mathcal{P}}
\renewcommand{\r}{r}
\newcommand{\R}{\mathfrak{R}}
\newcommand{\dom}{\mathrm{dom}}
\newcommand{\opt}{\mathrm{opt}}
\newcommand{\col}{\Omega}
\newcommand{\initmark}{\iota}

\newcommand{\cdelaygame}[2]{\Gamma\!_{#1}(#2)}
\newcommand{\SigmaI}{\Sigma_I}
\newcommand{\SigmaO}{\Sigma_O}
\newcommand{\game}{\mathcal{G}}

\newcommand{\exptime}{\textsc{ExpTime}}
\newcommand{\bigo}{\mathcal{O}}

\begin{document}

\maketitle
\begin{abstract}
We present an exponential-time algorithm approximating the minimal lookahead necessary to win an $\omega$-regular delay game. 
\end{abstract}

\section{Introduction}
\label{sec:intro}
Games can be found in the standard toolkit for many areas of theoretical computer science and mathematics, e.g., set theory and logic, automata theory, and complexity theory. 
Here, we are concerned with two-player zero-sum games of infinite duration and perfect information.
In their most abstract form, these are known as Gale-Stewart games~\cite{GS}, played between Player~$I$ and Player~$O$ in rounds~$i \in \nats$: In each round~$i$, first Player~$I$ picks some letter~$a_i$ from an alphabet~$\Sigma_I$, then Player~$O$ picks a letter~$b_i$ from an alphabet~$\SigmaO$. 
Thus, after $\omega$ rounds, they have produced an infinite sequence~$w = \binom{a_0}{b_0}\binom{a_1}{b_1}\binom{a_2}{b_2}\cdots$ of letters. 
Now, Player~$O$ wins such a play, if $w$ satisfies some given winning condition, e.g., $w \in L$ for some given set~$L$ of infinite words. 
In this work, we only consider $\omega$-regular~$L$ given by deterministic parity automata.

Note that here both players move in strict alternation and the automaton will process the sequence in this order. 
Hosch and Landweber introduced delay games, relaxing this rigid interaction by allowing Player~$O$ to delay her moves in order to obtain a lookahead on Player~$I$'s moves~\cite{HoschL72}.\footnote{We refer to the  introduction of~\cite{HoltmannKaiserThomas12} for a discussion of the history of delay games, including motivation and a connection to the uniformization of $\omega$-regular relations by continous functions.}

Hosch and Landweber proved that it is decidable whether Player~$O$ wins an $\omega$-regular delay game with some bounded lookahead.
In later work, Holtmann, Kaiser, and Thomas showed that such bounded lookahead is sufficient in the following sense: in an $\omega$-regular delay game, Player~$O$ either wins with doubly-exponential lookahead or not at all (not even with unbounded lookahead)~\cite{HoltmannKaiserThomas12}.
In subsequent work, an improved exponential upper bound and matching lower bounds have been proved~\cite{KleinZimmermann16}. 

Here, we consider the problem of determining the minimal lookahead that is sufficient for Player~$O$ to win an $\omega$-regular delay game. 
It is trivial to determine this value in doubly-exponential-time by hardcoding the exponential lookahead into the game, thereby turning the delay game into a classical, i.e., delay-free, game (see~\cite{WinterZimmermann20}, Section 3.1 for details).
As the resulting games can be solved in doubly-exponential time, one obtains the minimal lookahead in doubly-exponential time by exhaustive search.
However, this has to be contrasted with the $\exptime$-hardness of determining whether Player~$O$ wins with some lookahead~\cite{KleinZimmermann16}, the only known lower bound on the complexity of the optimization problem.

We present the first improvement over the naive algorithm for the lookahead optimization problem by presenting an exponential-time algorithm approximating the minimal lookahead within a factor of two. 
To this end, we show that the exponential-time algorithm for determining whether Player~$O$ wins for some lookahead can be refined into an approximation algorithm with an exponential running time. 
Due to the hardness result for the related decision problem, this is the best running time one can hope for (barring major surprises in complexity theory). 

\section{$\omega$-regular Delay Games}
\label{sec:delaygames}
Given an alphabet~$\Sigma$, i.e., a non-empty finite set of letters, $\Sigma^*$ and $\Sigma^\omega$ denote the set of finite respectively infinite words over $\Sigma$. 
Given a product alphabet~$\SigmaI \times \SigmaO$ we write $\binom{a_0a_1a_2 \cdots}{b_0b_1b_2\cdots}$ for the word~$\binom{a_0}{b_0}\binom{a_1}{b_1}\binom{a_2}{b_2}\cdots$ with $a_i\in \SigmaI$ and $b_i \in \SigmaO$. Also, we use similar notation for finite words, provided they are of the same length.
We denote the empty word by $\epsilon$, the power set of a set~$S$ by $\pow{S}$, and the set of non-negative integers by $\nats$.

A delay game (with constant lookahead) $\cdelaygame{k}{L}$ consists of a lookahead~$k \in\nats$ and a winning condition~$L \subseteq (\SigmaI \times \SigmaO)^\omega$.
It is played in rounds~$i \in \nats$ as follows: In round~$0$, Player~$I$ picks letters~$a_0 \cdots a_{k}$ from $\SigmaI$, then Player~$O$ picks a letter~$b_0$ from $\SigmaO$.
In round~$i>0$, Player~$I$ picks a letter~$a_{k+i} \in \SigmaI$ and then Player~$O$ picks a letter~$b_i \in \SigmaO$.
After $\omega$ rounds, they have produced an outcome~$\binom{a_0}{b_0}\binom{a_1}{b_1}\binom{a_2}{b_2} \cdots$. 
We say that the outcome is winning for Player~$O$ if it is in $L$.

A strategy for Player~$O$ is a mapping~$\sigma \colon \SigmaI^* \rightarrow \SigmaO$. 
An outcome~$\binom{a_0}{b_0}\binom{a_1}{b_1}\binom{a_2}{b_2} \cdots$ is consistent with $\sigma$, if $b_i = \sigma(a_0 \cdots a_{i+k})$ for all $i$.
A strategy~$\sigma$ is winning, if every outcome that is consistent with $\sigma$ is winning for Player~$O$.
If Player~$O$ has a winning strategy for $\cdelaygame{k}{L}$, then we say she wins~$\cdelaygame{k}{L}$.

\begin{remark}[\cite{HoltmannKaiserThomas12}, Remark 3.3]
\label{remark:mono}
If Player~$O$ wins $\cdelaygame{k}{L}$, then she also wins $\cdelaygame{k'}{L}$ for every $k' > k$.
\end{remark}
 
In this work, we consider winning conditions~$L$ recognized by deterministic parity automata (DPA)~$\aut = (Q, \Sigma, q_\initmark, \delta, \col)$, where $Q$ is a finite set of states containing the initial state~$q_\initmark$, $\Sigma$ is the input alphabet, $\delta\colon Q\times \Sigma \rightarrow Q$ is the transition function, and $\col\colon Q\rightarrow \nats$ is a coloring of the states. 
As usual, we extend $\delta$ to finite words by defining~$\delta^* \colon Q \times \Sigma^* \rightarrow Q$ via $\delta^*(q,\epsilon) = q$ and $\delta^*(q,wa) = \delta(\delta^*(q,w),a)$ for all $q \in Q$, $w \in \Sigma^*$, and $a \in \Sigma$. 
Given a word~$a_0 a_1 a_2 \cdots \in \Sigma^\omega$, the run of $\aut$ on $\alpha$ is the unique sequence~$q_0q_1 q_2 \cdots$ of states given by $q_i = \delta^*(q_\initmark, a_0 \cdots a_{i-1})$.
A run $q_0q_1 q_2 \cdots $ is accepting, if the largest color appearing infinitely often in the run is even, i.e., $\limsup_{i \rightarrow \infty} \col(q_i) \bmod 2 = 0$.
The language~$L(\aut)$ recognized by $\aut$ is the set containing all words whose run is accepting.

It is known that one can determine in exponential time whether Player~$O$ wins a delay game for some lookahead~$k$, if the winning condition~$L$ is recognized by a DPA. 

\begin{proposition}[\cite{KleinZimmermann16}, Theorem 4.4]
\label{prop:complexity}
The following problem is $\exptime$-complete: Given a DPA~$\aut$, does Player~$O$ win $\cdelaygame{k}{L(\aut)}$ for some $k$?
\end{proposition}

Furthermore, there is an exponential upper bound on the  lookahead necessary to win a delay game. 

\begin{proposition}[\cite{KleinZimmermann16}, Theorem 4.8]
\label{prop:upperbound}
Let $\aut$ be a DPA with $n$ states and $c$ colors, and define $k_{\max} = 2^{n^2c+1}$. Player~$O$ wins $\cdelaygame{k}{L(\aut)}$ for some $k$ if, and only if, she wins $\cdelaygame{k_{\max}}{L(\aut)}$.
\end{proposition}

Finally, the exponential upper bound on the necessary lookahead is tight.

\begin{proposition}[\cite{KleinZimmermann16}, Theorem 3.2]
\label{prop:lowerbound}
For every $n >1$, there is a language~$L_n$ recognized by a DPA with $\bigo(n)$ states and two colors such that Player~$O$ wins $\cdelaygame{k}{L_n}$ for some $k$, but she does not win $\cdelaygame{2^n}{L_n}$.
\end{proposition}

In this work, we consider the following problem: Given a DPA~$\aut$ over $\SigmaI \times \SigmaO$, determine the smallest $k$ such Player~$O$ wins $\cdelaygame{k}{L(\aut)}$ (or that there is no such $k$). Due to Proposition~\ref{prop:upperbound}, the search space for the smallest such $k$ is bounded by $k_{\max} = 2^{n^2c+1}$. 

Now, one can easily transform a game~$\cdelaygame{k}{L(\aut)}$ (for some fixed $k$) into an equivalent classical parity game (see, e.g.,~\cite{2001automata} for an introduction to parity games) encoding a queue of $k$ letters from $\SigmaI$ implementing the lookahead (\cite{WinterZimmermann20}, Section~3.1). 
Thus, one can construct the equivalent parity game for each $k \le k_{\max}$ and determine the smallest $k$ such that Player~$O$ wins the resulting parity game. 
This is also the smallest $k$ such that Player~$O$ wins $\cdelaygame{k}{L(\aut)}$).

However, if $k$ is exponential in the size of $\aut$ (i.e., close to $k_{\max}$), then the resulting parity game is of doubly-exponential size in the size of $\aut$, as one encodes a queue of exponential length. 
Due to Proposition~\ref{prop:lowerbound}, considering an exponential $k$ is, in general, unavoidable.
Hence, the resulting algorithm has doubly-exponential running time.

In the following, we show that the minimal lookahead can be approximated within a factor of two in exponential time.
As the related decision problem is $\exptime$-hard (Proposition~\ref{prop:complexity}), one cannot do better than exponential time (barring major surprises in complexity theory).

\section{The Algorithm}
\label{sec:algorithm}
Given a DPA~$\aut$ over $\SigmaI \times \SigmaO$ and a $k  > 0$, we define a game~$\game_k$ played between Player~$I$ and $O$ with the following properties:
\begin{enumerate}
    
    \item\label{property:delay2nodelay} If Player~$O$ wins $\cdelaygame{k}{L(\aut)}$, then she wins $\game_k$ (Lemma~\ref{lemma:correctness:delay2nodelay}).
    
    \item\label{property:nodelay2delay} If Player~$O$ wins $\game_k$, then she wins $\cdelaygame{2k-1}{L(\aut)}$ (Lemma~\ref{lemma:correctness:nodelay2delay}). 
    
    \item\label{property:complexity} Given $\aut$ and $k \le 2^{n^2c+1}$, one can construct $\game_k$  and determine its winner in exponential time in $n$, where $n$ and $c$ are the number of states and colors of $\aut$ (Lemma~\ref{lemma:runningtime}).

\end{enumerate}
Now, consider Algorithm~\ref{algorithm}.

\begin{algorithm}[h]
\begin{algorithmic}[1]
 \FOR {$k = 1 $ \TO $2^{n^2c+1}$}
 	\IF { Player~$O$ wins $\game_k$ }
 	\RETURN {$2k-1$}
 	\ENDIF
 \ENDFOR
  \RETURN {\myquot{Player~$O$ does not win for any lookahead~$k$}}
\end{algorithmic} 
\caption{Approximating the minimal lookahead necessary to win a delay game with winning condition~$L(\aut)$, where $\aut$ is a given DPA with $n$ states and $c$ colors.}
\label{algorithm}
\end{algorithm}

It is obvious that the algorithm runs in exponential time, as the calls in Line~2 can be executed in exponential time (Property~\ref{property:complexity}) and the loop terminates after an exponential number of iterations (which can obviously be reduced to a polynomial number of iterations using binary search (see Remark~\ref{remark:mono})). 

Now, fix an input $\aut$.  
If Player~$O$ does not win $\cdelaygame{k}{L(\aut)}$ for any $k$, then she does not win any $\game_k$ (Property~\ref{property:nodelay2delay}), i.e., the algorithm returns the correct output in Line~4.
So, consider the case where Player~$O$ wins $\cdelaygame{k}{L(\aut)}$ for some $k$.
Further, let $k_\opt$ be the minimal $k$ such that Player~$O$ wins $\cdelaygame{k}{L(\aut)}$. 
Proposition~\ref{prop:upperbound} yields $k_\opt \le 2^{n^2c+1}$. 
Hence, Player~$O$ wins $\game_{k}$ for some $k \le 2^{n^2c+1}$ due to Property~\ref{property:delay2nodelay}.
We pick $k^*$ minimal with this property, i.e., $2k^*-1$ is the output of the algorithm. 
Due to Property~\ref{property:nodelay2delay}, the output~$2k^*-1$ allows Player~$O$ to win the delay game with winning condition~$L(\aut)$.
Finally, the algorithm indeed approximates the minimal lookahead within a factor of two:
Due to Property~\ref{property:delay2nodelay}, we have $k^* \le k_\opt$, which implies that the approximation ratio between the algorithm's output~$2k^*-1$ and the optimal value~$k_\opt$ is indeed bounded by two:
\[
\frac{2k^*-1}{k_\opt} \le\frac{2k^*}{k_\opt} \le \frac{2k_\opt}{k_\opt} \le 2
.\]
Altogether, we obtain our main result.

\begin{theorem}
The following problem can be approximated within a factor of two in exponential time: Given a DPA~$\aut$, determine the smallest $k$ such that Player~$O$ wins $\cdelaygame{k}{L(\aut)}$.
\end{theorem}

Note that we do not consider the computation of a strategy realizing the approximation, as the notion of finite-state strategies for delay games comes with some technical complications~\cite{WinterZimmermann20}.

In the remainder of this section, we present the construction of $\game_k$ and prove Properties~\ref{property:delay2nodelay}, \ref{property:nodelay2delay}, and \ref{property:complexity}.

\subsection{The Game~$\game_k$}
\label{subsec:gamedef}
The construction of $\game_k$ is a refinement of a similar game used to prove Proposition~\ref{prop:complexity}~\cite{KleinZimmermann16,WinterZimmermann20}. 
The idea behind the construction is to define an equivalence relation over words in $\SigmaI^k$ that induce the same \emph{behavior} in the projection of $\aut$ to $\SigmaI$ (note that this automaton is nondeterministic due to the projection). Then, $\game_k$ is a delay-free game in which Player~$I$ picks such equivalence classes and Player~$O$ resolves the nondeterminism encoded in these classes, thereby producing a run of $\aut$. 
Finally, to account for the delay, Player~$I$ is one move ahead, i.e., he picks two equivalence classes before Player~$O$ makes her first move. 
For a detailed explanation of the construction, we refer the reader to \cite{KleinZimmermann16,WinterZimmermann20}. 

Throughout this section, we fix~$\aut = (Q, \SigmaI \times \SigmaO, q_\initmark, \delta, \col)$, some $0< k \le 2^{n^2c+1}$, and let $C = \col(Q)$  denote the set of colors of $\aut$.

First, we modify the transition function of $\aut$ so that it keeps track of the maximal color occurring along a (partial) run of $\aut$. Formally, we define $\delta_\C \colon (Q \times  C) \times (\SigmaI \times \SigmaO) \rightarrow (Q \times  C)$ for all $q \in Q$, $c \in  C$, and $\binom{a}{b} \in \SigmaI \times \SigmaO$ via
\[\delta_\C\left( (q,c),\binom{\footnotesize a}{\footnotesize b} \right) = ( q', \max\set{ c, \col(q') } ),\]
where $q' = \delta(q,\binom{a}{ b})$.

Next, we project away the $\SigmaO$-component of the letter and perform a power set construction by defining $\delta_\P \colon \pow{Q \times  C} \times \SigmaI \rightarrow \pow{Q \times  C}$ via
\[
\delta_\P(S, a) = \bigcup_{(q,c) \in S}\,\, \bigcup_{b\in \SigmaO} \delta_\C\left((q,c),\binom{\footnotesize a}{\footnotesize b}\right)
\]
for all $S \subseteq Q\times C$ and $a \in \SigmaI$.
We extend $\delta_\P$ to non-empty words via $\delta^*_\P(S, \epsilon) = S$ and $\delta^*_\P(S, wa) = \delta_\P(\delta^*_\P(S,w) ,a)$ for all $w \in \SigmaI^*$ and $a \in\SigmaI$.

Finally, for every non-empty $D \subseteq Q \times  C$ and $w \in \SigmaI^*$, we define the function~$\r_w^D \colon D\rightarrow \pow{Q \times  C} $ via
\[\r_w^D(q,c) = \delta_\P^*(\set{(q, \col(q))},w)\]
for all $(q,c) \in D$.
Note that the first argument of $\delta_\P^*$ is $(q,\col(q))$ and not $(q,c)$, the argument of $\r_w^D$, as we restart the process of keeping track of the maximal color of a run infix.
While the color~$c$ in the argument of $\r_w^D$ is therefore irrelevant, it simplifies our notation later on.

\begin{remark}
$(q',c') \in \r_w^D(q,c) $ if and only if there is a word~$w'$ over $\SigmaI \times \SigmaO$ whose projection to $\SigmaI$ is $w$ and such that the run of $\aut$ processing $w$ from $q$ leads to $q'$ and has maximal color~$c'$. 
\end{remark}

We call $w \in \SigmaI^k$ a witness for a partial function~$\r\colon Q \times  C \rightharpoonup \pow{Q \times  C}$ if we have $\r = \r^{\dom(r)}_w$, where $\dom(r)$ denotes the domain of $\r$. 
Note that we require a witness to have length~$k$.
Let $\R$ be the set of all such functions that have a witness, i.e.,
\[
\R = \set{\r^D_w \mid w \in \SigmaI^k \text{ and } D \subseteq Q \times C}.
\]

    
    
    


Now, we define $\game_k$, which is played in rounds~$i\in\nats$ between Player~$I$ and Player~$O$. In each round~$i$, Player~$I$ has to pick some $\r_i \in \R$ and then Player~$O$ picks $(q_i, c_i) \in Q \times  C$ subject to the following constraints:
\begin{itemize}
    \item For Player~$I$: $\dom(\r_0) = \set{(q_\initmark,\col(q_\initmark))}$ and  $\dom(\r_i) = \r_{i-1}(q_{i-1}, c_{i-1})$ for all $i > 0$.
    \item For Player~$O$: $(q_i,c_i) \in \dom(\r_i)$ for all $i \ge 0$. 
\end{itemize}
It is straightforward to verify that both players always have at least one move available in every round. 
A play of $\game_k$ is a sequence~$\r_0 (q_0, c_0) \r_1 (q_1, c_1) \r_2 (q_2, c_2)\cdots \in (\R \cdot (Q \times C))^\omega $. 
It is winning for Player~$O$ if the sequence of colors satisfies the parity condition, i.e., if $\limsup_{i\rightarrow \infty}c_i$ is even. 

A strategy~$\sigma$ for Player~$O$ maps a sequence~$\r_0 (q_0, c_0) \cdots \r_i$ to a pair~$(q_i, c_i) \in \dom(\r_i)$.
A play $\r_0 (q_0, c_0) \r_1 (q_1, c_1) \r_2 (q_2, c_2) \cdots$ is consistent with $\sigma$ if $(q_i, c_i) = \sigma(\r_0 (q_0, c_0) \cdots \r_i)$ for all $i \ge 0$. 
We say that $\sigma$ is a winning strategy for Player~$O$ if every play that is consistent with $\sigma$ is winning for her.
Finally, Player~$O$ wins $\game_k$ if she has a winning strategy.

\subsection{Correctness}
\label{subsec:correctness}
\begin{lemma}
\label{lemma:correctness:delay2nodelay} 
If Player~$O$ wins $\cdelaygame{k}{L(\aut)}$, then she wins $\game_k$.
\end{lemma}

\begin{proof}
Let $\sigma$ be a winning strategy for Player~$O$ in $\cdelaygame{k}{L(\aut)}$. 
We construct a winning strategy~$\sigma'$ for Player~$O$ in $\game_k$, which will simulate $\sigma$.

So, let $\r_0 \in \R$ be the first move of Player~$I$.
This has to be answered by Player~$O$ by picking $(q_0,c_0) = (q_\initmark, \col(q_\initmark))$, as this is the only legal move for her.
Hence, we define $\sigma'(\r_0) = (q_\initmark, \col(q_\initmark))$.
Now, Player~$I$ picks some $\r_1 \in \R$.

We simulate this in $\cdelaygame{k}{L(\aut)}$ as follows. Pick witnesses~$w_0$ and $w_1$ for $\r_0$ and $\r_1$, respectively. 
If Player~$I$ uses $w_0 w_1$ during the first $k$ rounds of $\cdelaygame{k}{L(\aut)}$, then $\sigma$ yields $k$ letters $w_0' \in \SigmaO^k$ as response. 
Thus, we are in the following situation for $i=1$:
\begin{itemize}
    \item In $\game_k$, we have a play prefix~$\r_0 (q_0, c_0) \cdots (q_{i-1}, c_{i-1})\r_i$, and
    \item in $\cdelaygame{k}{L(\aut)}$, Player~$I$ has picked $w_0\cdots w_i$ and Player~$O$ has picked $w_0' \cdots w_{i-1}'$, where each $w_{j}$ is a witness for $\r_j$ (and thus is in $\SigmaI^k$) and each $w_j'$ is in $\SigmaO^k$.
\end{itemize}

Now, consider an arbitrary~$i \ge 1$. Let $q_i$ be the unique state of $\aut$ that is reached from $q_{i-1}$ by processing $\binom{w_{i-1}}{w_{i-1}'}$, and let $c_i$ be the maximal color on this finite run infix from $q_{i-1}$ to $q_i$. 
We have $(q_i,c_i) \in \r_{i-1}(q_{i-1},c_{i-1})$, i.e., $(q_i,c_i)$ is a legal move for Player~$O$ in $\game_k$ to extend the play prefix~$\r_0 (q_0, c_0) \cdots (q_{i-1}, c_{i-1}) \r_i$. 
Accordingly, we define $\sigma'(\r_0 (q_0, c_0) \cdots(q_{i-1}, c_{i-1}) \r_i) = (q_i,c_i)$.
Player~$I$ reacts by picking some $\r_{i+1} \in \R$, which has some witness~$w_{i+1}$. 
In $\cdelaygame{k}{L(\aut)}$ we let Player~$I$ pick the letters of $w_{i+1}$ during the next $k$ rounds, which yield $k$ letters~$w_i' \in \SigmaO^k$ determined by $\sigma$. 
Thus, we are in the above situation for $i+1$, i.e., we have concluded the definition of $\sigma'$.

It remains to show that $\sigma'$ is indeed winning.
Fix a play~$\r_0 (q_0, c_0) \r_1 (q_{1}, c_{1}) \r_2 (q_2, c_2) \cdots$ that is consistent with $\sigma'$, and let $\binom{w_0 w_1 w_2 \cdots }{w_0' w_1' w_2' \cdots }$ be the play in $\cdelaygame{k}{L(\aut)}$ constructed during the simulation. 
By construction, each $w_i$ is a witness of $\r_i$. 
An induction shows that $q_{i+1}$ is the unique state of $\aut$ reached when processing $\binom{w_i}{w_i'}$ when starting at $q_i$, and that $c_{i+1}$ is the maximal color encountered on this run infix. 
As the unique run of $\aut$ on $\binom{w_0 w_1 w_2 \cdots }{w_0' w_1' w_2' \cdots }$ is accepting (as it is, by construction, an outcome consistent with the winning strategy~$\sigma$), we conclude that $\limsup_{i\rightarrow \infty}c_i$ is even. 
Hence, the play~$\r_0 (q_0, c_0) \r_1 (q_{1}, c_{1}) \r_2 (q_2, c_2) \cdots$ is winning for Player~$O$.
As the play was chosen arbitrarily, $\sigma'$ is indeed a winning strategy for Player~$O$ in $\game_k$. 
\end{proof}

\begin{lemma}
\label{lemma:correctness:nodelay2delay} 
If Player~$O$ wins $\game_k$, then she wins  $\cdelaygame{2k-1}{L(\aut)}$.
\end{lemma}

\begin{proof}
Let $\sigma'$ be a winning strategy for Player~$O$ in $\game_k$.
We construct a winning strategy~$\sigma$ for Player~$O$ in $\cdelaygame{2k-1}{L(\aut)}$, which will simulate~$\sigma'$.

So, let Player~$I$ pick letters $ a_0\cdots  a_{2k-1}$ in round~$0$ and define $w_0 =  a_0 \cdots  a_{k-1}$ and $w_1 =  a_k \cdots  a_{2k-1}$.
Furthermore, let $(q_0, c_0) = (q_\initmark, \col(q_\initmark))$, $\r_0 = \r_{w_0}^{\set{(q_0, c_0)}}$, and $\r_1 = \r_{w_1}^{\r_0(\set{(q_0, c_0)})}$. 
Then, $\r_0 (q_0, c_0) \r_1$ is a play prefix in $\game_k$ that is trivially consistent with $\sigma'$.
Then, we are in the following situation for $i = 1$:
\begin{itemize}
    \item In $\cdelaygame{2k-1}{L(\aut)}$, Player~$I$ has picked $w_0 \cdots w_i$ and Player~$O$ has picked $w_0' \cdots w_{i-2}'$ (which is empty for $i=1$), and
    \item in $\game_k,$ we have a play prefix~$\r_0 (q_0, c_0) \cdots (q_{i-1}, c_{i-1})\r_i$ that is consistent with $\sigma'$ and where each $w_j$ is a witness for $\r_j$. Note that being a play prefix implies $(q_j, c_j) \in \dom(\r_j) = \r_{j-1}(q_{j-1}, c_{j-1})$.
\end{itemize}

Now, pick some arbitrary $i \ge 1$ and consider $(q_i, c_i) = \sigma'(\r_0 (q_0, c_0) \cdots (q_{i-1}, c_{i-1})\r_i)$.
As $\sigma'$ is a strategy for $\game_k$, we have again $(q_i, c_i) \in \dom(\r_i) = \r_{i-1}(q_{i-1},c_{i-1})$.
Furthermore, as $w_{i-1}$ is a witness for $\r_{i-1}$, there is some $w_{i-1}' \in \SigmaO^k$ such that $q_i$ is the unique state $\aut$ reaches when processing $\binom{w_{i-1}}{w_{i-1}'}$ from $q_{i-1}$, and $c_i$ is the maximal color occurring in this run infix.

Now, we define~$\sigma$ such that it picks the $k$ letters of $w_{i-1}'$ during the next $k$ rounds (independently of the choices of Player~$I$). 
During these rounds, Player~$I$ again determines some $w_{i+1} \in \SigmaI^k$, inducing $\r_{i+1} = \r_{w_{i+1}}^{\r_i(q_i, c_i)}$. 
Then, we are in the above situation for $i+1$, i.e., we have concluded the definition of $\sigma$. 

It remains to show that $\sigma$ is winning. To this end, fix an outcome~$\binom{w_0w_1w_2 \cdots}{w_0'w_1'w_2'\cdots}$ that is consistent with $\sigma$, where each $w_i$ is in $\SigmaI^k$ and each $w_i'$ is in $\SigmaO^k$.
Further, let $\r_0(q_0, c_0)\r_1 (q_1,c_1)\r_2(q_2,c_2)\cdots$ be the play of $\game_k$ constructed during the simulation.
By construction, each $w_i$ is a witness of $\r_i$.

As $\r_0(q_0, c_0)\r_1 (q_1,c_1)\r_2(q_2,c_2)\cdots$ is consistent with $\sigma'$ by construction, it is winning for Player~$O$, i.e., $\limsup_{i\rightarrow \infty}c_i$ is even. 
Now, an induction shows that $q_{i+1}$ is the unique state reached by $\aut$ when processing~$\binom{w_i}{w_i'}$ starting in $q_i$, and $c_{i+1}$ is the maximal color on this run infix.
From these two properties, we conclude that the run of $\aut$ on $\binom{w_0w_1w_2 \cdots}{w_0'w_1'w_2'\cdots}$ is accepting, i.e., the outcome is winning for Player~$O$.
As the outcome was chosen arbitrarily, $\sigma$ is indeed a winning strategy for Player~$O$ in $\cdelaygame{2k-1}{L(\aut)}$.
\end{proof}

\subsection{Running Time}
\label{subsec:runningtime}
\begin{lemma}
\label{lemma:runningtime}
Given $\aut$ and $k \le 2^{n^2c+1}$, one can construct $\game_k$  and determine its winner in exponential time in $n$, where $n$ and $c$ are the number of states and colors of $\aut$.
\end{lemma}

\begin{proof}
We argue that $\game_k$ can be expressed as an arena-based parity  game (see, e.g., \cite{2001automata} for a definition) of exponential size in $n$ with the same colors as $\aut$.
Such a game can be solved in exponential time in $n$~\cite{parity}.
Thus, it remains to argue that one can construct the parity game in exponential time.

First, we argue that for each partial function~$\r \colon Q \times C \rightarrow \pow{Q \times C}$ one can construct a deterministic finite automaton recognizing the set of witnesses of $\r$. The construction is based on a powerset construction (mirroring the definition of $\delta_\C$ and $\delta_\P$) and a counter checking that only inputs of length~$k$ are accepted. 
As there are only exponentially many such functions, one can effectively determine $\R$, i.e., the set of functions whose associated automaton has a non-empty language, in exponential time.

Now, it is straightforward to construct a parity game~$(V_I, V_O, E, v_\initmark, \col')$ in a graph~$(V_I\cup V_O, E)$ of exponential size with the following components:
\begin{itemize}
    \item $V_I = \set{v_\initmark} \cup \R\times (Q \times   C)$: vertices of Player~$I$, where $v_\initmark$ is a fresh initial vertex.
    \item $V_O = \R$: vertices of Player~$O$.
    \item $E$ is the union of the following sets of edges:
    \begin{itemize}
        \item $\set{(v_\initmark, \r) \mid \dom(r) = \set{q_\initmark, \col(q_\initmark)}}$: initial moves of Player~$I$, allowing him to pick some $\r \in \R$ with $\dom(r) = \set{q_\initmark, \col(q_\initmark)}$.
        \item $\set{((\r, (q,c)), \r') \mid \dom(\r') = \r(q,c)}$: non-initial moves of Player~$I$ allowing him to pick some $r'$ satisfying $\dom(\r') = \r(q,c)$, where $\r$ and $(q,c)$ were previously picked by the players.
        \item $\set{ (r, (r, (q,c))) \mid (q,c) \in \dom(r) }$: moves of Player~$O$ allowing her to pick some~$(q,c) \in \dom(\r)$, where $\r$ was previously picked by Player~$I$.
    \end{itemize}
    \item $\col'(v) = \begin{cases}
    c & \text{if }v = (r,(q,c)) \in V_I,\\
    \min  C & \text{otherwise.}
    \end{cases}
    $ Note that the color~$\min C$ is neutral in the following sense: Whether a play is winning or not only depends on the colors of the vertices in $V_I\setminus\set{v_\initmark}$, but not on vertices in $V_O\cup\set{v_\initmark}$.  
\end{itemize}
The resulting parity game implements exactly the rules of the abstract game~$\game_k$ and is therefore won by Player~$O$ if and only if she wins $\game_k$. 
\end{proof}

\section{Conclusion}
\label{sec:conc}
We have presented an exponential-time algorithm approximating the minimal lookahead necessary to win a delay game. 
Here, we only considered the case of $\omega$-regular winning conditions given by deterministic parity automata.

In the literature, several other types of winning conditions have been considered, e.g., quantitative parity~\cite{Zimmermann17} and (quantitative) Linear Temporal Logic~\cite{KleinZ16}. 
For these types, one can also exhibit an approximation algorithm for the minimal lookahead that has the same complexity as an algorithm deciding the existence of some lookahead using techniques very similar to those introduced here.

Unfortunately, the complexity of the exact optimization problem for games with winning conditions given by deterministic parity automata remains open. 
Let us conclude by mentioning another open problem on delay games: There is an exponential gap between the upper and lower bounds on the necessary lookahead in delay games with winning conditions given by deterministic Muller automata. 
The same is true for deciding whether Player~$O$ wins the game for some lookahead.

\bibliographystyle{plain}

\bibliography{sample.bib}

\end{document}